\documentclass{article}
\usepackage{nips14submit_e,times}

\newcommand{\hf}{\hat{f}}

\newcommand{\hy}{\hat{y}}

\newcommand{\hbeta}{\hat{\beta}}

\newcommand{\hmu}{\hat{\mu}}
\newcommand{\hpi}{\hat{\pi}}

\newcommand{\nn}{\mathcal{N}}
\newcommand{\RR}{\mathbb{R}}

\newcommand{\EE}[2][]{\mathbb{E}_{#1}\left[#2\right]}

\newcommand{\Var}[2][]{\operatorname{Var}_{#1}\left[#2\right]}
\newcommand{\hVar}[2][]{\widehat{\operatorname{Var}}_{#1}\left[#2\right]}

\newcommand{\hCov}[2][]{\widehat{\operatorname{Cov}}_{#1}\left[#2\right]}
\newcommand{\Norm}[1]{\left\lVert #1\right\rVert}

\newcommand{\argmin}{\operatorname{argmin}}

\newcommand{\eqand}{\text{ and }}

\newcommand{\where}{\text{ where }}
\newcommand{\with}{\text{ with }}

\newcommand{\simiid}{\stackrel{iid}{\sim}}
\newcommand{\p}[1]{\left(#1\right)}

\usepackage{amsmath, amssymb, amsthm}
\usepackage{graphicx}
\usepackage{caption, subcaption}
\usepackage[margin=1.5in]{geometry}
\usepackage{hyperref}

\newtheorem{thm}{Theorem}

\graphicspath{{.}}

\let\emptyset\varnothing

\nipsfinalcopy % Uncomment for camera-ready version

\newcommand{\fbb}[1]{{\scriptstyle [#1]}}
\newcommand{\yold}{\hy_i^{(t)}}
\newcommand{\pred}{\check{y}_i^{(t)}}
\newcommand{\ynew}{\hy_i^{(t+1)}\fbb{\yold}}
\newcommand{\ynewplain}{\hy_i^{(t+1)}}
\newcommand{\ynewclean}{\hy_i^{(t+1)}\fbb{\emptyset}}
\newcommand{\ynewp}{\hy_i^{(t+1)}\fbb{\pred}}
\newcommand{\plainy}{y_i^{(t)}}
\newcommand{\ynewy}{\hy_i^{(t+1)}\fbb{\plainy}}
\newcommand{\yoldclean}{\hy_i^{(t)}\fbb{\emptyset}}
\newcommand{\noise}{\nu_i^{(t)}}
\newcommand{\ynewnoised}{\hy_i^{(t+1)}\fbb{\yold + \noise}}
\newcommand{\xnew}{x_i^{(t+1)}\fbb{\yold}}

\newcommand{\xnewclean}{x_i^{(t+1)}\fbb{\emptyset}}

\newcommand{\noisesd}{\sigma_\nu}
\newcommand{\noisevar}{\noisesd^2}

\usepackage{verbatim}

\newcommand{\condbar}{{\ \Big| \ }}
\newcommand{\noiseden}{\varphi_N} %{\varphi_{\noisesd}}

\newcommand{\secshrink}{\vspace{-1mm}}

\title{Feedback Detection for Live Predictors}

\makeatletter
\renewcommand*{\@fnsymbol}[1]{\ensuremath{\ifcase#1\or  \or \dagger\or \ddagger\or
   \mathsection\or \mathparagraph\or \|\or **\or \dagger\dagger
   \or \ddagger\ddagger \else\@ctrerr\fi}}
\makeatother

\author{Stefan Wager, \, Nick Chamandy,\, Omkar Muralidharan,\, and\, Amir Najmi \\
\texttt{swager@stanford.edu, \{chamandy, omuralidharan, amir\}@google.com} \\
\textit{Stanford University and Google, Inc.}}
\begin{document}

\maketitle

\begin{abstract}
A predictor that is deployed in a live production system may perturb the features it uses to make predictions. Such a feedback loop can occur, for example, when a model that predicts a certain type of behavior ends up causing the behavior it predicts, thus creating a self-fulfilling prophecy. In this paper we analyze predictor feedback detection as a causal inference problem, and introduce a local randomization scheme that can be used to detect non-linear feedback in real-world problems. We conduct a pilot study for our proposed methodology using a predictive system currently deployed as a part of a search engine.
\end{abstract}

\section{Introduction}

When statistical predictors are deployed in a live production environment, feedback loops can become a concern. Predictive models are usually tuned using training data that has not been influenced by the predictor itself; thus, most real-world predictors cannot account for the effect they themselves have on their environment.
Consider the following caricatured example: A search engine wants to train a simple classifier that predicts whether a search result is ``newsy'' or not, meaning that the search result is relevant for people who want to read the news. This classifier is trained on historical data, and learns that high click-through rate (CTR) has a positive association with ``newsiness.'' Problems may arise if the search engine deploys the classifier, and starts featuring search results that are predicted to be newsy for some queries: promoting the search result may lead to a higher CTR, which in turn leads to higher newsiness predictions, which makes the result be featured even more.

If we knew beforehand all the channels through which predictor feedback can occur, then detecting feedback would not be too difficult.  For example, in the context of the above example, if we knew that feedback could \emph{only} occur through some changes to the search result page that were directly triggered by our model,
then we could estimate feedback by running small experiments where we turn off these triggering rules. However, in large industrial systems where networks of classifiers all feed into each other, we can no longer hope to understand {\it a priori} all the ways in which feedback may occur. We need a method that lets us detect feedback from sources we might not have even known to exist.

This paper proposes a
%fully automatic
simple method for detecting feedback loops from unknown sources in live systems.
%that does not require modeling the feedback mechanism.
Our method relies on artificially inserting a small amount of noise into the predictions made by a model, and then measuring the effect of this noise on future predictions made by the model. If future model predictions change when we add artificial noise,
%to the current ones,
then our system has feedback.

To understand how random noise can enable us to detect feedback, suppose that we have a model with predictions $\hy$ in which tomorrow's prediction $\hy^{(t+1)}$ has a linear feedback dependence on today's prediction $\hy^{(t)}$: if we increase $\hy^{(t)}$ by $\delta$, then $\hy^{(t+1)}$ increases by $\beta\,\delta$ for some $\beta \in \RR$. Intuitively, we should be able to fit this slope $\beta$ by perturbing $\hy^{(t)}$ with a small amount of noise $\nu \sim \nn\p{0, \, \noisevar}$ and then regressing the new $\hy^{(t+1)}$ against the noise; the reason least squares should work here is that the noise $\nu$ is independent of all other variables by construction. The main contribution of this paper is to turn this simple estimation idea into a general procedure that can be used to detect feedback in realistic problems where the feedback has non-linearities and jumps.

 %In Sections \ref{sec:methods} and \ref{sec:experiments} we show how to use our ideas in practice, and use our method to estimate the magnitude of feedback in a predictive model currently in use at an internet company.

\secshrink

\paragraph{Counterfactuals and Causal Inference}

Feedback detection is a problem in causal inference. A model suffers from feedback if the predictions it makes today affect the predictions it will make tomorrow. We are thus interested in discovering a causal relationship between today's and tomorrow's predictions; simply detecting a correlation is not enough. The distinction between causal and associational inference is acute in the case of feedback: today's and tomorrow's predictions are almost always strongly correlated, but this correlation by no means implies any causal relationship.

In order to discover causal relationships between consecutive predictions, we need to use some form of randomized experimentation. In our case, we add a small amount of random noise to our predictions. Because the noise is fully artificial, we can reasonably ask counterfactual questions of the type: ``How would tomorrow's predictions have changed if we added more/less noise to the predictions today?'' The noise acts as an independent instrument that lets us detect feedback.
We frame our analysis in terms of a potential outcomes model that asks how the world would have changed had we altered a treatment; in our case, the treatment is the random noise we add to our predictions. This formalism, often called the Rubin causal model \cite{holland1986statistics}, is regularly used for understanding causal inference \cite{angrist1996identification,efron1991compliance, imbens1994identification}.
Causal models are useful for studying the behavior of live predictive systems on the internet, as shown by, e.g., the recent work of Bottou et al. \cite{bottou2013counterfactual} and Chan et al. \cite{chan2010evaluating}.

\secshrink

\paragraph{Outline and Contributions}

In order to define a rigorous feedback detection procedure, we need to have a precise notion of what we mean by feedback. Our first contribution is thus to provide such a model by defining statistical feedback in terms of a potential outcomes model (Section \ref{sec:main}). Given this feedback model, we propose a local noising scheme that can be used to fit feedback functions with non-linearities and jumps (Section \ref{sec:details}). Before presenting general version of our approach, however, we begin by discussing the linear case in Section \ref{sec:linear} to elucidate the mathematics of feedback detection: as we will show, the problem of linear feedback detection using local perturbations reduces to linear regression. Finally, in Section \ref{sec:real_ex} we conduct a pilot study based on a predictive model currently deployed as a part of a search engine.

%We begin in Section \ref{sec:main} by defining statistical feedback in terms of a potential outcomes model. Once we have a firm handle on what we mean by feedback, the mathematical problem of feedback detection falls within the scope of classical statistics. As we verify in Section \ref{sec:linear}, the simplest case, where feedback is a linear function of past predictions, in fact reduces to linear regression; in Section \ref{sec:details} we generalize our approach to non-linear feedback functions. Finally, in Section \ref{sec:real_ex} we conduct a pilot study based on a predictive model currently in use at an internet company.

\secshrink

\section{Feedback Detection for Statistical Predictors}
\label{sec:main}

\secshrink

Suppose that we have a model that makes predictions $\yold$ in time periods $t = 1, \, 2, \, ...$ for examples $i = 1, \, ..., \, n$. The predictive model itself is taken as given; our goal is to understand feedback effects between consecutive pairs of predictions $\yold$ and $\ynewplain$.
%\paragraph{Feedback and potential outcomes}
We define statistical feedback in terms of counterfactual reasoning: we want to know what would have happened to $\ynewplain$ had $\yold$ been different. We use potential outcomes notation \cite{rubin2005causal} to distinguish between counterfactuals: let $\ynewy$ be the predictions our model  \emph{would have made} at time $t+1$ if we had published $\plainy$ as our time-$t$ prediction. In practice we only get to observe $\ynewy$ for a single $\plainy$; all other values of $\ynewy$ are counterfactual.
We also consider $\ynewclean$, the prediction our model would have made at time $t+1$ if the model never made any of its predictions public and so did not have the chance to affect its environment. With this notation, we define feedback as
\begin{equation}
%\label{eq:feedback}
\text{feedback}_i^{(t)} = \ynew - \ynewclean,
\end{equation}
i.e., the difference between the predictions our model actually made and the predictions it would have made had it not had the chance to affect its environment by broadcasting predictions in the past. Thus, statistical feedback is a difference in potential outcomes.

\secshrink

\paragraph{An additive feedback model}
In order to get a handle on feedback as defined above, we assume that feedback enters the model additively:
%\begin{equation}
%\label{eq:add}
$\ynewy = \ynewclean + f(\plainy),$
%\end{equation}
where $f$ is a feedback function, and $\plainy$ is the prediction published at time $t$. In other words, we assume that the predictions made by our model at time $t+1$ are the sum of the prediction the model would have made if there were no feedback, plus a feedback term that only depends on the previous prediction made by the model. Our goal is to estimate the feedback function $f$.

\secshrink

\paragraph{Artificial noising for feedback detection}  The relationship between $\yold$ and $\ynewplain$ can be influenced by many things, such as trends, mean reversion, random fluctuations, as well as feedback. In order to isolate the effect of feedback, we need to add some noise to the system to create a situation that resembles a randomized experiment.
Ideally, we might hope to sometimes turn our predictive system off in order to get estimates of $\yoldclean$. However, predictive models are often deeply integrated into large software systems, and it may not be clear what the correct system behavior would be if we turned the predictor off. To side-step this concern, we randomize our system by adding artificial noise to predictions: at time $t$, instead of deploying the prediction $\yold$, we deploy
%\begin{equation}
%\label{eq:pred}
$\pred = \yold + \noise$, where $\noise \simiid N$
%\end{equation}
is artificial noise drawn from some distribution $N$.
Because the noise $\noise$ is independent from everything else, it puts us in a randomized experimental setup that allows us to detect feedback as a causal effect. If the time $t+1$ prediction $\ynewplain$ is affected by $\noise$, then our system must have feedback because the only way $\noise$ can influence $\ynewplain$ is through the interaction between our model predictions and the surrounding environment at time $t$.

\secshrink

\paragraph{Local average treatment effect} In practice, we want the noise $\noise$ to be small enough that it does not disturb the regular operation of the predictive model too much. Thus, our experimental setup allows us to measure feedback as a local average treatment effect \cite{imbens1994identification}, where the artificial noise $\noise$ acts as a continuous treatment. Provided our additive model holds, we can then piece together these local treatment effects into a single global feedback function $f$.

\secshrink

\section{Linear Feedback}
\label{sec:linear}

\secshrink

We begin with an analysis of linear feedback problems; the linear setup allows us to convey the main insights with less technical overhead. We discuss the non-linear case in Section \ref{sec:details}.
Suppose that we have some natural process $x^{(1)}, \, x^{(2)}, \, ...$ and a predictive model of the form $\hy = w \cdot x$. (Suppose for notational convenience that $x$ includes the constant, and the intercept term is folded into $w$.) For our purposes, $w$ is fixed and known; for example, $w$ may have been set by training on historical data. At some point, we ship a system that starts broadcasting the predictions $\hy = w \cdot x$, and there is a concern that the act of broadcasting the $\hy$ may perturb the underlying $x^{(t)}$ process. Our goal is to detect any such feedback.
Following earlier notation we write
$\ynew = w \cdot \xnew $
for the time $t+1$ variables perturbed by feedback, and
$\ynewclean = w \cdot \xnewclean$
for the counterparts we would have observed without any feedback.

In this setup,  any effect of $\yold$ on $\xnew$ is feedback. A simple way to constrain this relationship is using a linear model
%\begin{equation}
%\label{eq:linx}
$\xnew = \xnewclean + \yold \, \gamma$.
%\end{equation}
In other words, we assume that $\xnew$ is perturbed by an amount that scales linearly with $\yold$. Given this simple model, we find that:
\begin{equation}
\label{eq:linear}
\ynew = \ynewclean + w \cdot \gamma \, \yold, 
\end{equation}
and so $f(y) = \beta \, y$ with $\beta = w \cdot \gamma$; $f$ is the feedback function we want to fit.

We cannot work with \eqref{eq:linear} directly, because $\ynewclean$ is not observed. In order to get around this problem, we add artificial noise to our predictions: at time $t$, we publish predictions $\pred = \yold + \noise$ instead of the raw predictions $\yold$. As argued in Section \ref{sec:main}, this method lets us detect feedback because $\ynewplain$ can only depend on $\noise$ through a feedback mechanism, and so any relationship between $\ynewplain$ and $\noise$ must be a symptom of feedback.

\secshrink

\paragraph{A Simple Regression Approach}

With the linear feedback model \eqref{eq:linear}, the effect of $\noise$ on $\ynewplain$ is
%\begin{equation}
%\label{eq:noiseeffect}
$\ynewnoised = \ynew + \beta \, \noise.$
%\end{equation}
This relationship suggests that we should be able to recover $\beta$ by regressing $\ynewplain$ against the added noise $\noise$. The following result confirms this intuition.

\begin{thm}
\label{thm:linear}
Suppose that \eqref{eq:linear} holds, and that we add noise $\noise$ to our time $t$ predictions. If we estimate $\beta$ using linear least squares
\begin{align}
\label{eq:simplefit}
\hbeta = \frac{\hCov{\ynewnoised, \, \noise}}{\hVar{\noise}}, \text{ then }
%\label{eq:simplebeta}
\sqrt{n} \, \p{\hbeta - \beta} \Rightarrow \nn\p{0, \, \frac{\Var{\ynew}}{\noisevar}},
\end{align}
where $\noisevar = \Var{\noise}$ and $n$ is the number of examples to which we applied our predictor.
\end{thm}

Theorem \ref{thm:linear} gives us a baseline understanding for the difficulty of the feedback detection problem: the precision of our feedback estimates scales as the ratio of the artificial noise $\noisevar$ to natural noise $\operatorname{Var}[\ynew]$. Note that the proof of Theorem \ref{thm:linear} assumes that we only used predictions from a single time period $t+1$ to fit feedback, and that the raw predictions $\ynew$ are all independent. If we relax these assumptions we get a regression problem with correlated errors, and need to be more careful with technical conditions.

\secshrink

\paragraph{Efficiency and Conditioning}
\label{sec:cond}

The simple regression model \eqref{eq:simplefit} treats the term $\ynew$ as noise. This is quite wasteful: if we know $\yold$ we usually have a fairly good idea of what $\ynew$ should be, and not using this information needlessly inflates the noise. Suppose that we knew the function\footnote{In practice we do not know $\mu$, but we can estimate it; see Section \ref{sec:details}.}
\begin{equation}
\label{eq:mu1}
\mu(y) := \EE{\ynew \condbar \yold = y}.
\end{equation}
Then, we could write our feedback model as
\begin{equation}
\ynewnoised = \mu\p{\yold} + \p{\ynew - \mu\p{\yold}} + \beta \, \noise,
\end{equation}
where $\mu(\yold)$ is a known offset. Extracting this offset improves the precision of our estimate for $\hbeta$.

\begin{thm}
\label{thm:lin_efficient}
Under the conditions of Theorem \ref{thm:linear} suppose that the function $\mu$ from \eqref{eq:mu1} is known and that the $\ynewplain$ are all independent of each other conditional on $\yold$. Then, given the information available at time $t$, the estimate
\begin{align}
\label{eq:eff_fit}
&\hbeta^* = \frac{\hCov{\ynewnoised - \mu\p{\yold}, \, \noise}}{\hVar{\noise}} \;\; \text{ has asymptotic distribution} \\
\label{eq:eff_beta}
&\sqrt{n} \, \p{\hbeta^* - \beta} \Rightarrow \nn\p{0, \, \frac{\EE{\Var{\ynew \condbar \yold}}}{\noisevar}}.
\end{align}
Moreover, if the variance of
$\eta_i^{(t)} = \ynew - \mu(\yold)$
does not depend on $\yold$, then $\hbeta^*$ is the best linear unbiased estimator of $\beta$.
\end{thm}

Theorem \ref{thm:lin_efficient} extends the general result from above that the precision with which we can estimate feedback scales as the ratio of artificial noise to natural noise.  The reason why $\hbeta^*$ is more efficient than $\hbeta$ is that we managed to condition away some of the natural noise, and reduced the variance of our estimate for $\beta$ by
\begin{equation}
\Var{\mu\p{\yold}} = \Var{\ynew} - \EE{\Var{\ynew \condbar \yold}}.
\end{equation}
In other words, the variance reduction we get from $\hbeta^*$ directly matches the amount of variability we can explain away by conditioning.
The estimator \eqref{eq:eff_fit} is not practical as stated, because it requires knowledge of the unknown function $\mu$ and is restricted to the case of linear feedback. In the next section, we generalize this estimator into one that does not require prior knowledge of $\mu$ and can handle non-linear feedback.

\secshrink

\section{Fitting Non-Linear Feedback}
\label{sec:details}

\secshrink

Suppose now that we have the same setup as in the previous section, except that now feedback has a non-linear dependence on the prediction:
 $\ynew = \ynewclean + f(\yold)$ for some arbitrary function $f$.
For example, in the case of a linear predictive model $\hy = w \cdot x$, this kind of feedback could arise if we have feature feedback
$ \xnew = \xnewclean + f_{(x)}(\yold); $
the feedback function then becomes $f(\cdot) = w \cdot f_{(x)}(\cdot)$.
When we add noise $\noise$ to the above predictions, we only affect the feedback term $f(\cdot)$:
\begin{equation}
\label{eq:delta}
\ynewnoised - \ynew = f\left(\yold + \noise\right) - f\left(\yold\right).
\end{equation}
Thus, by adding artificial noise $\noise$, we are able to cancel out the nuisance terms, and isolate the feedback function $f$ that we want to estimate.
We cannot use \eqref{eq:delta} in practice, though, as we can only observe one of $\ynewnoised$ or $\ynew$ in reality; the other one is counterfactual. We can get around this problem by conditioning on $\yold$ as in Section \ref{sec:cond}. Let
\begin{align}
\label{eq:mean_noise}
\mu\p{y}
&= \EE{\ynewnoised \condbar \yold = y} \\
\notag
&= t\p{y} + \varphi_{N} * f\p{y}, \where t\p{y} = \EE{\ynewclean \condbar \yold = y}
\end{align}
is a term that captures trend effects that are not due to feedback. The $*$ denotes convolution:
\begin{equation}
\noiseden * f\p{y} = \EE{f\p{\yold + \noise} \condbar \yold = y} \with \noise \sim N.  %\nn\p{0, \, \noisevar}.
\end{equation}
Using the conditional mean function $\mu$ we can write our expression of interest as
\begin{align}
\label{eq:main_regr}
\ynewnoised - \mu\p{\yold} = f\p{\yold + \noise} - \noiseden * f\p{\yold} + \eta_i^{(t)},
\end{align}
where $\eta_i^{(t)} := \ynewclean - t\p{\yold}$.
If we have a good idea of what $\mu$ is, the left-hand side can be measured, as it only depends on $\ynewnoised$ and $\yold$. Meanwhile, conditional on $\yold$, the first two terms on the right-hand side only depend on $\noise$, while $\eta_i^{(t)}$ is independent of $\noise$ and mean-zero. The upshot is that we can treat \eqref{eq:main_regr} as a regression problem where $\eta_i^{(t)}$ is noise.
In practice, we estimate $\mu$ from an auxiliary problem
%\begin{equation}
%\label{eq:mu_est}
where we regress $\ynewnoised$ against $\yold$.
%\end{equation}
%where we again used the notation $a \sim g(b)$ to mean that we want to learn a function $g(b)$ that predicts $a$.

\secshrink

\paragraph{A Pragmatic Approach}
\label{sec:pragmatic}

There are many possible approaches to solving the non-parametric system of equations \eqref{eq:main_regr} for $f$, e.g., \cite{hastie2009elements}, Chapter 5. Here, we take a pragmatic approach, and constrain ourselves to solutions of the form
%\begin{equation}
%\label{eq:constraint}
$\hmu(y) = \hbeta_\mu \cdot b_\mu(y)$ and $\hf(y) = \hbeta_f \cdot b_f(y)$,
%\end{equation}
where $b_\mu: \RR \rightarrow \RR^{p_\mu}$ and $b_f:\RR \rightarrow \RR^{p_f}$ are predetermined basis expansions. This approach transforms our problem into an ordinary least-squares problem, and works well in terms of producing reasonable feedback estimates in real-world problems (see Section \ref{sec:real_ex}). %In Section \ref{sec:penalized} we suggest a more general approach as a topic for further research.
If this relation in fact holds for some values $\beta_\mu$ and $\beta_f$, the result below shows that we can recover $\beta_f$ by least-squares.
%We provide a proof of this result, as well as a more detailed outline for how to set up the linear regression problem, in Appendix \ref{sec:nl_fit}.

\begin{thm}
\label{thm:spline}
Suppose that $\beta_\mu$ and $\beta_f$ are defined as above, and that we have an unbiased estimator $\hbeta_\mu$ of $\beta_\mu$ with variance
$ V_\mu = \operatorname{Var}[\hbeta_\mu]. $
Then, if we fit $\beta_f$  by least squares using \eqref{eq:main_regr} as described in Appendix \ref{sec:nl_fit}, the resulting estimate $\hbeta_f$ is unbiased and has variance
\begin{equation}
\label{eq:full_var}
\Var{\hbeta_f} = \p{X_f^\intercal X_f}^{-1}X_f^\intercal \p{V_Y + X_\mu V_\mu X_\mu^\intercal} X_f \p{X_f^\intercal X_f}^{-1},
\end{equation}
where the design matrices $X_\mu$ and $X_f$ are defined as
\begin{equation}
\label{eq:design}
X_\mu = \begin{pmatrix} \vdots \\ b^\intercal_\mu\p{\yold} \\ \vdots \end{pmatrix}
\eqand
X_f = \begin{pmatrix} \vdots \\ b^\intercal_f\p{\yold + \noise} - \p{\noiseden * b_f}^\intercal\p{\yold}\\ \vdots \end{pmatrix}
\end{equation}
and $V_Y$ is a diagonal matrix with $\p{V_Y}_{ii} = \Var{\ynew \condbar \yold}$.
\end{thm}

 In the case where our spline model is misspecified, we can obtain a similar result using methods due to Huber \cite{huber1967behavior} and White \cite{white1980heteroskedasticity}.
In practice, we can treat $\hbeta_\mu$ as known since fitting $\mu(\cdot)$ is usually easier than fitting $f(\cdot)$: estimating $\mu(\cdot)$ is just a smoothing problem whereas estimating $f(\cdot)$ requires fitting differences. If we also treat the errors $\eta_i^{(t)}$ in \eqref{eq:main_regr} as roughly homoscedatic, \eqref{eq:full_var} reduces to
\begin{align}
\label{eq:simple_var}
\Var{\hbeta_f} \approx \frac{\EE{\Var{\ynew \condbar \yold}}}{n \, \EE{\|s_i\|_2^2}},
\where s_i = b_f\p{\yold + \noise} - \noiseden * b_f\p{\yold}.
\end{align}
This simplified form again shows that the precision of our estimate of $f(\cdot)$ scales roughly as the ratio of the variance of the artificial noise $\noise$ to the variance of the natural noise.

\secshrink

\paragraph{Our Method in Practice}
%\label{sec:methods}

%If we are able to insert artificial noise into our predictions and willing to make the additive assumption, then we can fit general non-linear feedback functions $f$ by performing two carefully designed spline regressions. We postpone a more thorough theoretical analysis of our methods until Sections \ref{sec:linear} and \ref{sec:details}; this section shows how to fit feedback in practice. In Appendix \ref{sec:nl_fit} we show how to carry out the steps below in \texttt{R}.
For convenience, we summarize the steps needed to implement our method here:
{\bf (1)} At time $t$, compute model predictions $\yold$ and draw noise terms $\noise \simiid N$ for some noise distribution $N$.
Deploy predictions $\pred = \yold + \noise$ in the live system.
{\bf (2)} Fit a non-parametric least-squares regression of $\ynewnoised \sim \mu\left(\yold\right)$ to learn the function
%\begin{equation}
%\label{eq:mean_trend}
$\mu\left(y \right) := \EE{\ynewnoised \condbar \yold = y}$.
%\end{equation}
We use the \texttt{R} formula notation, where $a \sim g(b)$ means that we want to learn a function $g(b)$ that predicts $a$. 
{\bf (3)} Set up the non-parametric least-squares regression problem
\begin{equation}
\label{eq:main}
\ynewnoised - \mu\left(\yold\right) \sim f\left(\yold + \noise\right) - \varphi_{N} * f\left(\yold\right),
\end{equation}
where the goal is to learn $f$. Here, $\varphi_{N}$ is the density of $\noise$, and $*$ denotes convolution. %A motivation for \eqref{eq:main} is given in Section \ref{sec:details}.
%\end{enumerate}
In Appendix \ref{sec:nl_fit} we show how to carry out these steps using standard \texttt{R} libraries.

The resulting function $f(y)$ is our estimate of feedback: If we make a prediction $\pred$ at time $t$, then our time $t+1$ prediction will be boosted by $f(\pred)$. The above equation only depends on $\yold$, $\noise$, and $\ynewnoised$, which are all quantities that can be observed in the context of an experiment with noised predictions. Note that as we only fit $f$ using the differences in \eqref{eq:main}, the intercept of $f$ is not identifiable. We fix the intercept (rather arbitrarily) by setting the average fitted feedback over all training examples to 0; we do not include an intercept term in the basis $b_f$.

\secshrink

\paragraph{Choice of Noising Distribution}

Adding noise to deployed predictions often has a cost that may depend on the shape of the noise distribution $N$. A good choice of $N$ should reflect this cost. For example, if the practical cost of adding noise only depends on the largest amount of noise we ever add, then it may be a good idea to draw $\noise$ uniformly at random from $\{\pm \varepsilon\}$ for some $\varepsilon > 0$. In our experiments, we draw noise from a Gaussian distribution $\noise \sim \nn(0, \, \noisevar)$. %; we write $\noiseden$ for the corresponding noise density in \eqref{eq:main}.

\begin{comment}

\paragraph{Non-parametric regression} A simple way to fit the non-parametric regressions \eqref{eq:mean_trend} and \eqref{eq:main} is to constrain both $\mu$ and $f$ to be functions of the form
\begin{equation}
\label{eq:basis}
\mu(y) = \beta_\mu \cdot b_\mu(y) \eqand f(y) = \beta_f \cdot b_f(y),
\end{equation}
where $b_\mu, \, b_f : \RR \rightarrow \RR^{p_\mu}, \, \RR^{p_f}$ are some basis expansions of our choice; the goal is then to learn $\beta_\mu$ and $\beta_f$. With this construction, \eqref{eq:mean_trend} and \eqref{eq:main} become $p_\mu$- and $p_f$-dimensional linear regression problems that can be solved using standard software such as \texttt{lm} in \texttt{R}. As shown in Theorem \ref{thm:spline}, this approach consistently recovers the best possible feedback function of the form \eqref{eq:basis}.

\paragraph{Choice of basis}
In order to get useful estimates of $f$, we need to choose the basis functions $b_\mu(y)$ and $b_f(y)$ in \eqref{eq:basis} carefully. In our applications, we focused on classification problems where $\hy$ is a log-odds estimate
$$ \hy = \log\left[\frac{\hpi}{1 - \hpi}\right]$$
for the probability $\hpi$ of some event of interest. For such examples, we built both $b_\mu(y)$ and $b_f(y)$ as natural splines with knots uniformly spread out on $\hy \in [-3, \, 3]$, as well as a jump at $\hy = 0$. The reason for these choices is that we are most interested in fitting feedback near the even-odds decision boundary $\hy = 0$. Other applications way require different choices of $b$.

\end{comment}

\secshrink

\section{A Pilot Study}
\label{sec:real_ex}

\secshrink

The original motivation for this research was to develop a methodology for detecting feedback in real-world systems. Here, we present results from a pilot study, where we added signal to historical data that we believe should emulate actual feedback. The reason for monitoring feedback on this system is that our system was about to be more closely integrated with other predictive systems, and there was a concern that the integration could induce bad feedback loops. Having a reliable method for detecting feedback would provide us with an early warning system during the integration.
%We provide a more extensive simulation study in Appendix \ref{sec:experiments}.

The predictive model in question is a logistic regression classifier. We added feedback to historical data collected from log files according to half a dozen rules of the form ``if $a_i^{(t)}$ is high and $\pred > 0$, then increase $a_i^{(t+1)}$ by a random amount''; here $\pred$ is the time-$t$ prediction deployed by our system (in log-odds space) and $a_i^{(t)}$ is some feature with a positive coefficient. These feedback generation rules do not obey the additive assumption. Thus our model is misspecified in the sense that there is no function $f$ such that a current prediction $\pred$ increased the log-odds of the next prediction by $f(\pred)$, and so this example can be taken as a stretch case for our method.

\begin{figure}[t]
\vspace{-5mm}
\centering
\includegraphics[width=0.5\textwidth]{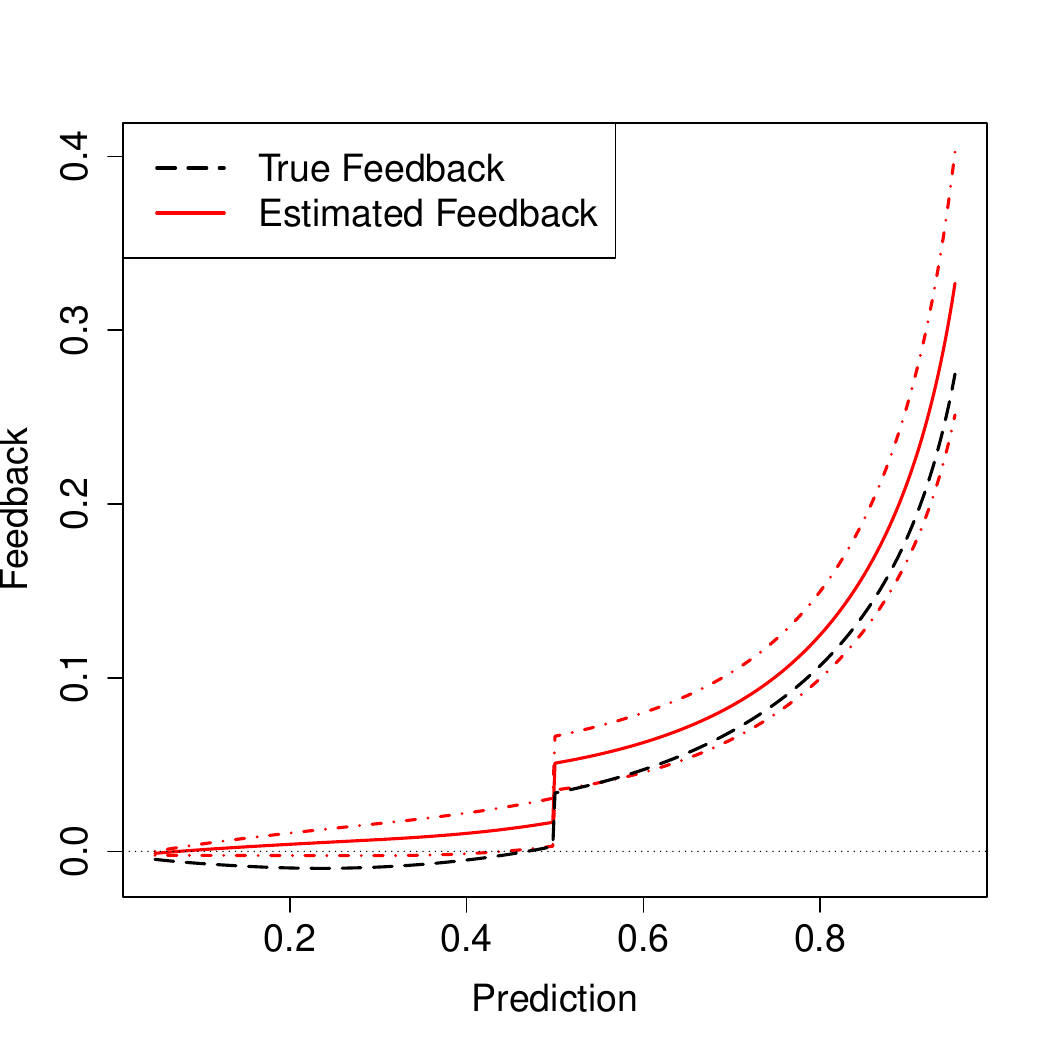}
\vspace{-3mm}
\caption{Simulation aiming to replicate realistic feedback in a real-world classifier. The red solid line is our feedback estimate; the black dashed line is the best additive approximation to the true feedback. The $x$-axis shows predictions in probability space; the $y$ axis shows feedback in log-odds space. The error bars indicate pointwise confidence intervals obtained using a non-parametric bootstrap with $B = 10$ replicates, and stretch 1 SE in each direction. Further experiments are provided in Appendix \ref{sec:experiments}.}
\vspace{-3mm}
\label{fig:real_sim}
\end{figure}

Our dataset had on the order of 100,000 data points, half of which were used for fitting the model itself and half of which were used for feedback simulation. We generated data for 5 simulated time periods, adding noise with $\noisesd = 0.1$ at each step, and fit feedback using a spline basis discussed in Appendix \ref{sec:experiments}.
The ``true feedback" curve was obtained by fitting a spline regression to the additive feedback model by looking at the unobservable $\ynewclean$; we used a $df = 5$ natural spline with knots evenly spread out on $[-9, 3]$ in log-odds space plus a jump at 0.

%For the real classifier we want to understand, we have fairly strong reasons to believe that if the feedback function has jumps, then there should be a jump at the even odds threshold.
For our classifier of interest, we have fairly strong reasons to believe that the feedback function may have a jump at zero, but probably shouldn't have any other big jumps. Assuming that we know {\it a priori} where to look for jumps does not seem to be too big a problem for the practical applications we have considered. Results for feedback detection are shown in Figure \ref{fig:real_sim}. Although the fit is not perfect, we appear to have successfully detected the shape of feedback. The error bars for estimated feedback were obtained using a non-parametric bootstrap \cite{efron1993introduction} for which we resampled pairs of (current, next) predictions. %It appears that in case of model misspecification, the parametric error bars that worked well in the simulation studies from Figure \ref{fig:sim} break down and become anti-conservative; this is why we resorted to using a bootstrap here.

This simulation suggests that our method can be used to accurately detect feedback on scales that may affect real-world systems. Knowing that we can detect feedback is reassuring from an engineering point of view.  On a practical level, the feedback curve shown in Figure \ref{fig:real_sim} may not be too big a concern \emph{yet}: the average feedback is well within the noise level of the classifier. But in large-scale systems the ways in which a model interacts with its environment is always changing, and it is entirely plausible that some innocuous-looking change in the future would increase the amount of feedback. Our methodology provides us with a way to continuously monitor how feedback is affected by changes to the system, and can alert us to changes that cause problems. In Appendix \ref{sec:experiments}, we show some simulations with a wider range of effect sizes.

\vspace{-3mm}

\section{Conclusion}

\vspace{-2.5mm}

In this paper, we proposed a randomization scheme that can be used to detect feedback  in real-world predictive systems. Our method involves adding noise to the predictions made by the system; this noise puts us in a randomized experimental setup that lets us measure feedback as a causal effect. In general, the scale of the artificial noise required to detect feedback is smaller than the scale of the natural predictor noise; thus, we can deploy our feedback detection method without disturbing our system of interest too much. The method does not require us to make hypotheses about the mechanism through which feedback may propagate, and so it can be used to continuously monitor predictive systems and alert us if any changes to the system lead to an increase in feedback.

\vspace{-2.5mm}

\paragraph{Related Work}

The interaction between models and the systems they attempt to describe has been extensively studied across many fields. Models can have different kinds of feedback effects on their environments. At one extreme of the spectrum, models can become self-fulfilling prophecies: for example, models that predict economic growth may in fact cause economic growth by instilling market confidence \cite{merton1948self,ferraro2005economics}. At the other end, models may distort the phenomena they seek to describe and therefore become invalid. A classical example of this is a concern that any metric used to regulate financial risk may become invalid as soon as it is widely used, because actors in the financial market may attempt to game the metric to avoid regulation \cite{danielsson2002emperor}.
However, much of the work on model feedback in fields like finance, education, or macro-economic theory has focused on negative results: there is an emphasis on understanding when feedback can happen and promoting awareness about how feedback can interact with policy decisions, but there does not appear to be much focus on actually fitting feedback.
%This emphasis on negative results is understandable as causal inference is difficult, especially at the scale of whole societies.
%The problem we address in this paper lets us get stronger results. In our setup, we are able to closely control the interaction between our model and the outside world, and conduct randomized experiments by perturbing our model. This enables us to get positive results, and to accurately estimate the effect of feedback.
%We are not aware of much prior work in terms of quantitative feedback detection.
One notable exception is a paper by Akaike \cite{akaike1968use}, who showed how to fit cross-component feedback in a system with many components; however, he did not add artificial noise to the system, and so was unable to detect feedback of a single component on itself.
%Finally, we note that this paper only dealt with detecting feedback; a related interesting problem is to remove feedback.

\vspace{-2.5mm}

\paragraph{Acknowledgments}

The authors are grateful to Alex Blocker, Randall Lewis, and Brad Efron for helpful suggestions and interesting conversations. S. W. is supported by a B. C. and E. J. Eaves Stanford Graduate Fellowship.

\newpage

{\small
\bibliographystyle{unsrt}
\bibliography{references}
}

\newpage

\begin{appendix}

\section{Fitting Non-Linear Feedback by Ordinary Least Squares Regression}
\label{sec:nl_fit}

\secshrink

Carrying out the fitting procedure outlined in Section \ref{sec:nl_fit} is straight-forward using standard \texttt{R} functions if we are willing to construct $\mu(\cdot)$ and $f(\cdot)$ using pre-specified basis expansions
\begin{equation}
\hmu(y) = \hbeta_\mu \cdot b_\mu(y) \eqand \hf(y) = \hbeta_f \cdot b_f(y). 
\end{equation}
Recall that $b_f$ cannot have an intercept, as it would not be identifiable. We first need to construct the design matrices
\begin{equation}
X_\mu = \begin{pmatrix} \vdots \\ b^\intercal_\mu\p{\yold} \\ \vdots \end{pmatrix}
\eqand
X_f = \begin{pmatrix} \vdots \\ b^\intercal_f\p{\yold + \noise} - \p{\noiseden * b_f}^\intercal\p{\yold}\\ \vdots \end{pmatrix} 
\end{equation}
from \eqref{eq:design}. Constructing $X_\mu$ just involves choosing a basis function; however, evaluating
\begin{equation}
\gamma_i = \p{\noiseden * b_f}^\intercal\p{\yold} 
\end{equation}
for each row of $X_f$ can be computationally intensive if we are not careful. In particular, evaluating $\gamma_i$ by numerical integration separately for each $i$ can be painfully slow. A more efficient way to compute $\gamma_i$ is to evaluate $\p{\noiseden * b_f}(y)$ over a grid of $y$-values in a single pass using the fast Fourier transform (e.g., by using \texttt{convolve} in \texttt{R}), and then to linearly interpolate the result onto the real line (e.g., using \texttt{approxfun}).

Once we have computed these design matrices, we can estimate $\hbeta_\mu$ and $\hbeta_f$ by solving the linear regression problems
\begin{equation}
\label{eq:mu_reg}
Y \sim X_\mu \hbeta_\mu
\end{equation}
and
\begin{equation}
\label{eq:f_reg}
\p{Y - X_\mu \hbeta_\mu} \sim X_f \hbeta_f,
\end{equation}
where $Y$ is just a vector with entries $\ynewnoised$. Notice that this whole procedure only requires knowledge of $X_\mu$, $X_f$, and the noised new predictions $\ynewnoised$; we never reference the counterfactual predictions $\ynew$ or unobservable predictions $\ynewclean$.

In practice, most of the errors in our procedure come from the difference equation \eqref{eq:f_reg} and not from the conditional mean regression \eqref{eq:mu_reg}. Thus, when our model is well-specified and the additivity assumption holds, we can get good estimates for the accuracy of $f$ by looking at the parametric standard error estimates provided by \texttt{lm} from fitting \eqref{eq:f_reg}; this is what we did for the simulations presented in Figure \ref{fig:sim}. In case of model misspecification, however, parametric confidence intervals can break down and it is better to use non-parametric methods such as the bootstrap. We used a non-parametric bootstrap for the logs simulation presented in Section \ref{sec:real_ex}.

\secshrink

\section{Simulation Experiments}
\label{sec:experiments}

\secshrink

Here, we present a collection of simulation experiments, the results of which are given in Figure \ref{fig:sim}. These examples are all logistic regression examples with additive feedback in log-odds space. In the plots, the $y$-axis shows feedback in log-odds space, whereas the $x$-axis shows deployed predictions in probability space.

\begin{figure}[p]
\vspace{-5mm}
        \centering
        \begin{subfigure}[b]{0.4\textwidth}
                \includegraphics[width=\textwidth]{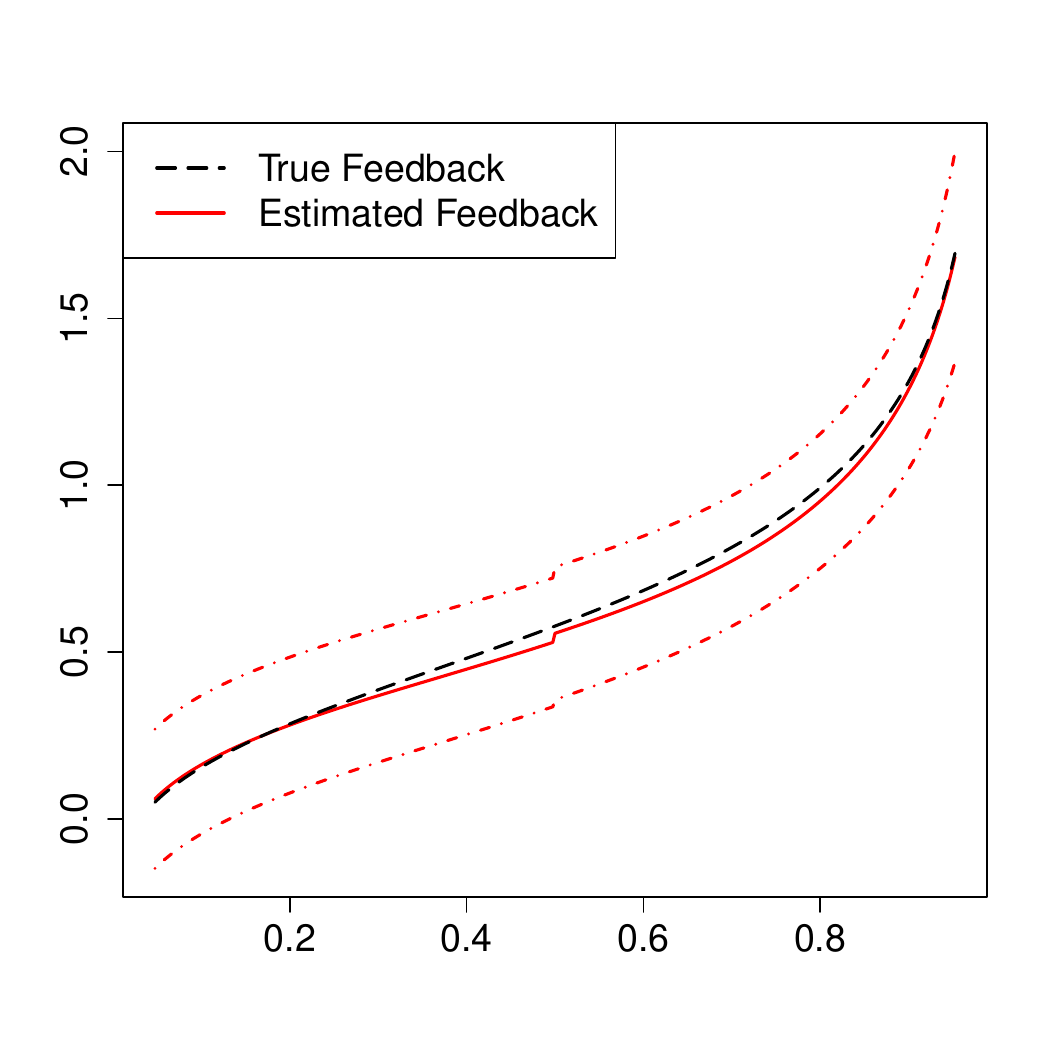}
                \caption{Continuous, monotone feedback}
        \end{subfigure}%
	\hspace{0.05\textwidth}
        \begin{subfigure}[b]{0.4\textwidth}
                \includegraphics[width=\textwidth]{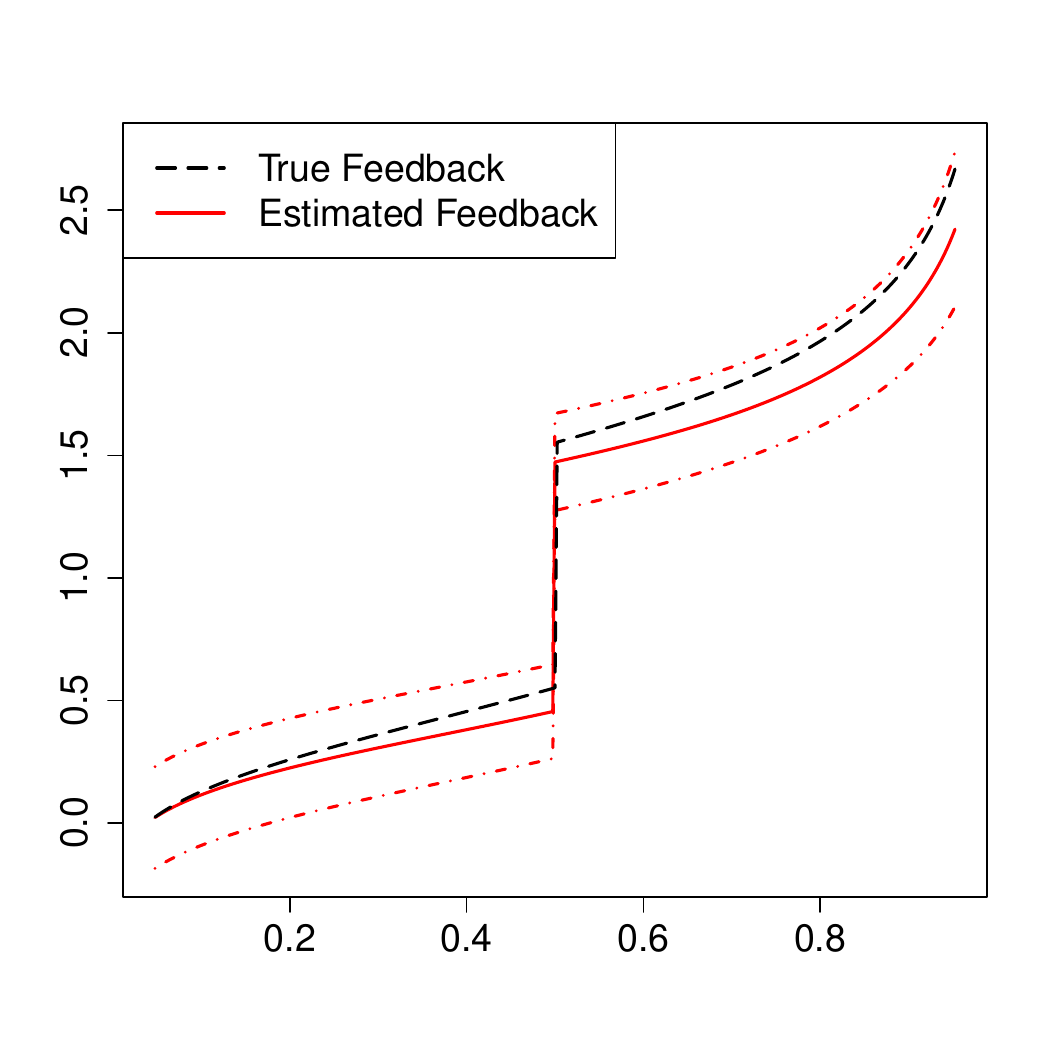}
                \caption{Monotone feedback with jump}
        \end{subfigure}%
	\\
\vspace{-3mm}
        \begin{subfigure}[b]{0.4\textwidth}
                \includegraphics[width=\textwidth]{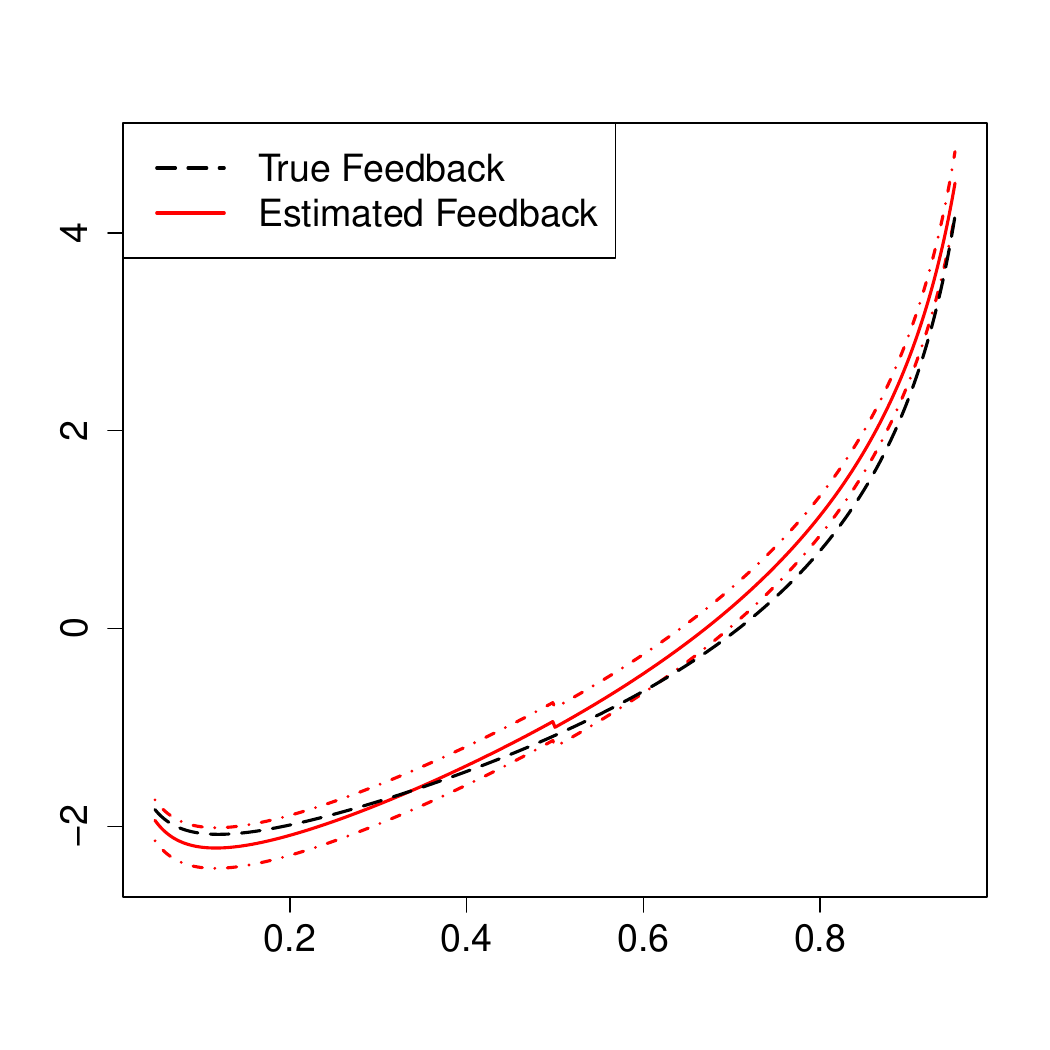}
                \caption{Continuous, non-monotone feedback}
        \end{subfigure}%
	\hspace{0.05\textwidth}
        \begin{subfigure}[b]{0.4\textwidth}
                \includegraphics[width=\textwidth]{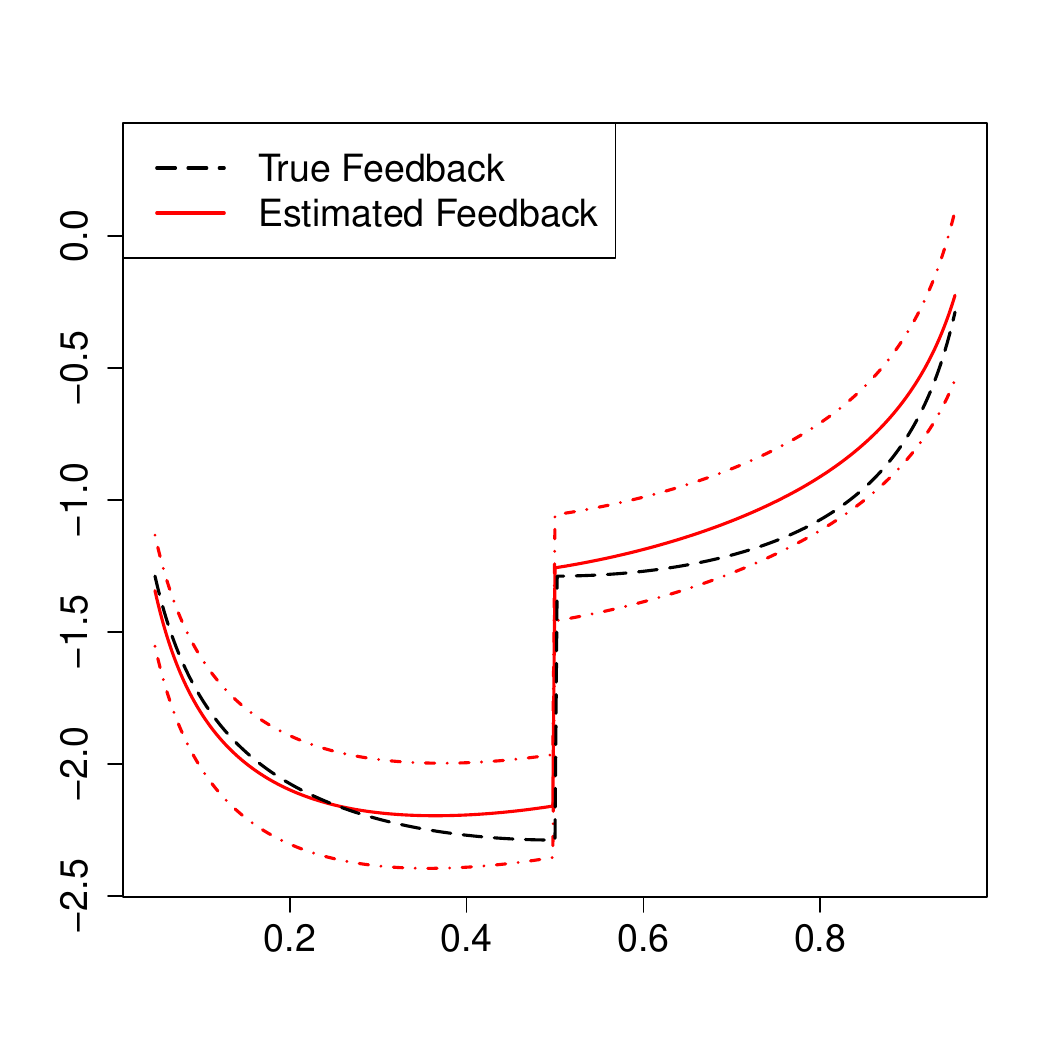}
                \caption{Non-monotone feedback with jump}
        \end{subfigure}%
	\\
\vspace{-3mm}
        \begin{subfigure}[b]{0.4\textwidth}
                \includegraphics[width=\textwidth]{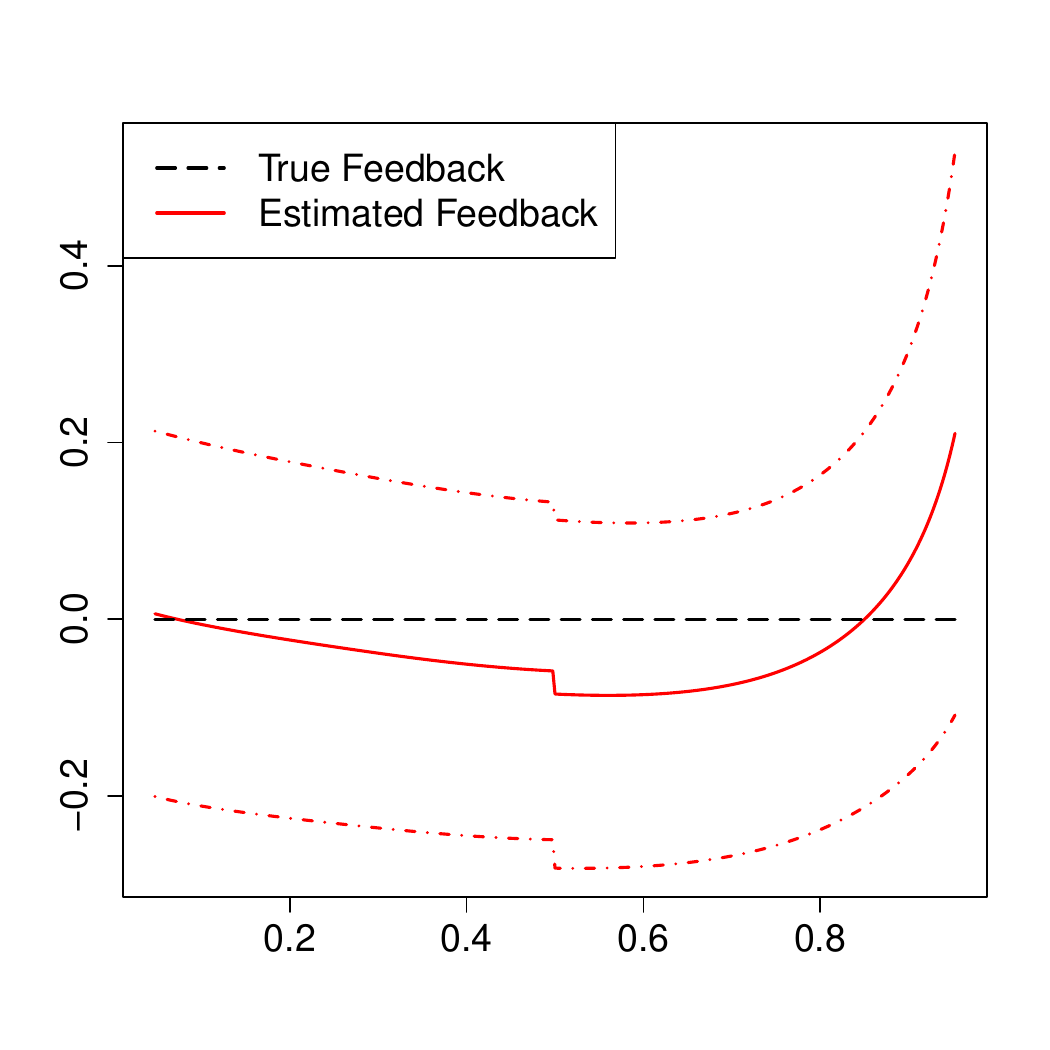}
                \caption{No feedback}
        \end{subfigure}%
	\hspace{0.05\textwidth}
        \begin{subfigure}[b]{0.4\textwidth}
                \includegraphics[width=\textwidth]{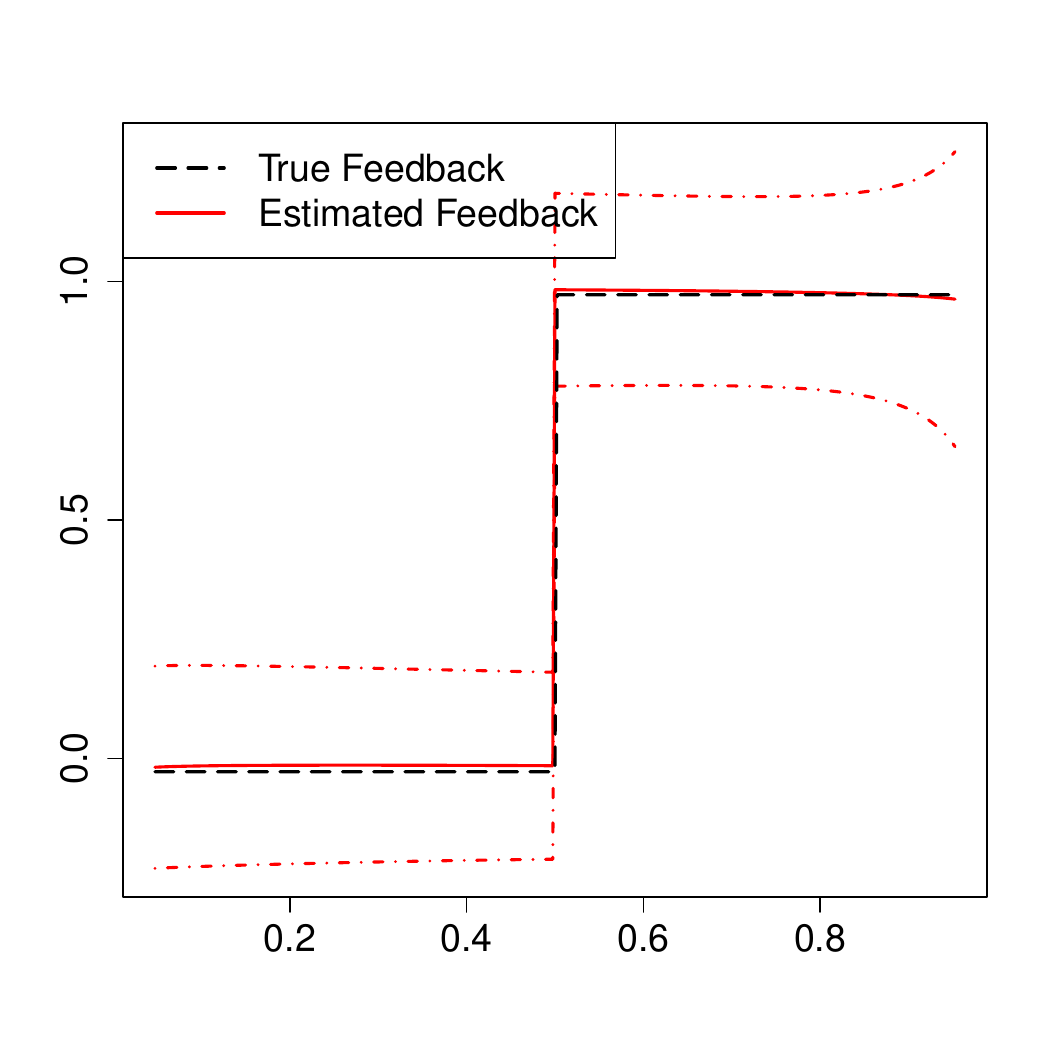}
                \caption{Jump only}
        \end{subfigure}%
\vspace{-3mm}
\caption{Testing the proposed feedback detection method on some simulation examples. We plot actual predictions in probability space on the $x$-axis against feedback in log-odds space on the $y$-axis. The dashed black line is the true feedback; the solid red line is our feedback estimate along with point-wise error bars stretching 1 SE in each direction. Note that in panel (e) the $y$-axis has a much finer scale than in the other panels.}
\label{fig:sim}
\end{figure}

The simulations all had $n = 100,000$ (old prediction, new prediction) pairs. The predictions had natural noise with standard error $\sigma = 0.5$, i.e., the $i^{th}$ pair was centered at $\mu_i$ and distributed as $\yold, \, \ynewclean \simiid \nn\p{\mu_i, \, 0.5^2}$. We added Gaussian noise with $\noisesd = 0.25$ to the deployed predictions.
In order to mimic real datasets, we made our simulation highly imbalanced: There were many strong predictions for the negative class with $\mu_i \ll 0$, but less so for the positive class. This is why our model performed better near $x = 0$ than near $x = 1$.

We fit both the trend $\mu(\cdot)$ and the feedback function $f(\cdot)$ as the sum of a natural spline with $df = 3$ degrees of freedom and knots spread evenly over $[-3, \, 3]$, and a jump at zero log-odds (i.e., $x = 0.5$). The dashed lines show the different feedback functions $f$ used in each example.

As emphasized earlier, the intercept of the feedback function $f$ is not identifiable from our experiments. We fixed the intercept by setting the average fitted feedback over all training examples to 0. Since all our training sets were heavily imbalanced, this effectively amounted to setting feedback to 0 at $x = 0$. {The plots that do not hit the (0, 0) point are missing a sharp spike at the left-most end; the plot ends at $x = \text{logit}(-3) \approx 0.05 $.}

As we see from Figure \ref{fig:sim}, our method accurately fits the feedback function in all six examples, including the null case with no feedback. The error bars depict standard asymptotic error bars produced by the \texttt{R} function \texttt{lm} when fitting \eqref{eq:main}.

%\begin{comment}

\secshrink

\section{Extensions and Further Work}

\secshrink
 
In this section, we discuss some possible extensions to the work presented in this paper.

\secshrink

\subsection{Feedback Removal}

\secshrink

If we detect feedback in a real-world system, we can try to identify the root causes of the feedback and fix the problem by removing the feedback loop. That being said, a natural follow-up question to our research is whether we can automatically remove feedback.
In the context of the linear feedback model \eqref{eq:linear}, we incur an expected squared-error loss of
$$ \tilde{\ell} = \beta^2 \, \EE{\p{\yold}^2} $$
from completely ignoring the feedback problem. Meanwhile, if we use the maximum likelihood estimate $\hbeta$ to correct feedback, we suffer a loss
$$ \ell_{\noisesd} = \Var{\hbeta} \, \EE{\p{\yold}^2} + \noisevar, $$
where the first term comes from our errors in estimating $\hbeta$ and the second comes from the extra noise we needed to inject into the system in order to detect the feedback.

An interesting topic for further research would be to find how to optimally set the scale $\noisesd$ of the artificial noise under various utility assumptions, and to understand the potential failure modes of feedback removal under model misspecification. In order to remove feedback, we would also need to have some way of dealing with the intercept term.

\secshrink

\subsection{Covariate-Dependent Feedback}

\secshrink

Our analysis was presented in the context of the additive feedback model
$$ \ynewp = \ynewclean + f\p{\pred}. $$
In practice, however, we may want to let feedback depend on some other covariates $z$
$$ \ynewp = \ynewclean + f\p{\pred, \, z_i^{(t)}}; $$
for example, we may want to slice feedback by geographic region. One particularly interesting but challenging extension would be to make feedback depend on the unperturbed prediction $\yoldclean$:
$$ \ynewp = \ynewclean + f\p{\pred, \, \yoldclean}. $$
For example, if $\hy$ is a prediction for how good a search result is, we might assume that search results that are actually good $(\hy\fbb{\emptyset} \gg 0)$ have a different feedback response from those that are terrible $(\hy\fbb{\emptyset} \ll 0)$. The challenge here is that $\hy\fbb{\emptyset}$ is unobserved, and so we need to have it act on $f$ via proxies. Developing a formalism that lets $f$ depend on $\hy\fbb{\emptyset}$ in a useful way while allowing for consistent estimation seems like a promising pathway for further work.

\secshrink

\subsection{Penalized Regression}
\label{sec:penalized}

\secshrink

The key technical challenge in implementing our method for feedback detection is solving the spline equation \eqref{eq:main_regr}. In Section \ref{sec:pragmatic} we proposed a pragmatic approach that enabled us to get good feedback estimates in many examples. However, it should be possible to devise more general methods for fitting $f$. The equation \eqref{eq:main_regr} is linear in $f$, and so any strictly convex penalty function $L: \mathcal{A} \rightarrow \RR$ over some convex subset $\mathcal{A} \subseteq \{\RR \rightarrow \RR\}$ of real valued functions on $\RR$ leads to a well-defined estimator $\hf$ through the convex optimization problem
\begin{align}
\label{eq:obj}
\hf_L = \argmin_{f \in \mathcal{A}} \Bigg\{\sum &\bigg(\ynewnoised - \mu\p{\yold} -  f\p{\yold + \noise}  \\
\notag
&\ \ \ \ \ + \noiseden * f\p{\yold}\bigg)^2 + L(f)\Bigg\}.
\end{align}
In the context of smoothing splines, a popular choice is to use
$$ L(f) = \lambda \int_\RR \Norm{f''(x)}^2 \ dx $$
and make $\mathcal{A}$ be the set on which this integral is well-defined. There is an extensive literature on non-parametric regression problems constrained by smoothness penalties
\cite{evgeniou2000regularization,girosi1995regularization,green1994nonparametric,hastie1990generalized,wahba1990spline}; presumably, similar approaches should also give us smoothing spline solutions to \eqref{eq:main_regr}.

\secshrink

\subsection{Deterministic Designs}

\secshrink

Finally, in this paper, we have focused on detecting feedback by adding random noise to raw model predictions. It would be interesting to see whether we can improve the efficiency of our procedure by optimizing the noise choice more closely and using a deterministic design. The problem of finding optimal designs for spline-type problems has been studied by several authors \cite{biedermann2011optimal,muller1996optimal,studden1969admissible}.

%\end{comment}

\secshrink

\section{Proofs}

\secshrink

\begin{proof}[Proof of Theorem \ref{thm:linear}]
Because $\noise$ is fully artificial noise, we know a-priori that $\noise$ and $\ynew$ are independent. Thus, we can treat $\ynew$ as a homoscedastic noise term for our regression, and \eqref{eq:simplefit} follows immediately from standard results for ordinary least squares regression.
\end{proof}

\secshrink

\begin{proof}[Proof of Theorem \ref{thm:lin_efficient}]
The $\eta_i^{(t)}$ are independent of the $\noise$, and so \eqref{eq:eff_beta} follows from an argument analogous to the one that led to \eqref{eq:simplefit}.
If the $\eta_i^{(t)}$ are still homoscedastic after conditioning on $\yold$ then, because the $\eta_i^{(t)}$ are mean-zero by construction, the fact that $\hbeta^*$ is the best linear unbiased estimator of $\beta$ follows directly from an application of the Gauss-Markov theorem where we treat $\noise$ as fixed and $\eta_i^{(t)}$ as random, see\cite{lehmann1998theory}, p. 184.
\end{proof}

\secshrink

\begin{proof}[Proof of Theorem \ref{thm:spline}]
Given the regression problem described above, \eqref{eq:full_var} follows directly from standard results on heteroscedastic linear regression \cite{huber1967behavior, white1980heteroskedasticity}. Note that our theoretical result assumes that $\hbeta_\mu$ and $\hbeta_f$ are trained on independent data sets.
\end{proof}

\end{appendix}

\end{document}